\documentclass[11pt,a4paper,english]{article}
\usepackage[T1]{fontenc}
\usepackage[margin=1in]{geometry}
\usepackage{amsmath}
\usepackage{amsthm}
\usepackage{amstext}
\usepackage{amssymb}
\usepackage{algorithm}
\usepackage[noend]{algpseudocode}
\usepackage{thmtools}
\usepackage{thm-restate}
\usepackage[hidelinks]{hyperref}

\usepackage{cleveref}
\crefname{theorem}{Theorem}{Theorems}
\crefname{observation}{Observation}{Observations}
\crefname{lemma}{Lemma}{Lemmas}
\crefname{corollary}{Corollary}{Corollaries}
\crefname{proposition}{Proposition}{Propositions}
\crefname{example}{Example}{Examples}
\crefname{claim}{Claim}{Claims}
\crefname{table}{Table}{Tables}
\crefname{equation}{Inequality}{Inequalities}
\crefname{reductionrule}{Reduction rule}{Reduction rules}
\crefname{section}{Section}{Sections}
\crefname{figure}{Figure}{Figures}
\crefname{chapter}{Chapter}{Chapters}

{\begin{list}{}%
         {\setlength{\leftmargin}{#1}}%
         \item[]%
}
{\end{list}}

\numberwithin{equation}{section}
\numberwithin{figure}{section}

\theoremstyle{plain}
\newtheorem{theorem}{Theorem}[section]

\newtheorem{lemma}[theorem]{Lemma}
\newtheorem{proposition}[theorem]{Proposition}
\newtheorem{claim}[theorem]{Claim}

\theoremstyle{remark}
\newtheorem{remark}[theorem]{Remark}

\theoremstyle{definition}
\newtheorem{definition}[theorem]{Definition}

\newcommand{\match}{\mathsf{Match}}
\newcommand{\UM}{\mathsf{UM^*}}
\newcommand{\UMs}{\mathsf{Uniform}^*}
\newcommand{\UMb}{\mathsf{Uniform_{bid}}}
\newcommand{\UMa}{\mathsf{Uniform_{ask}}}
\newcommand{\MM}{\mathsf{MM}}

\newcommand{\ts}{\mathsf{timestamp}}
 \newcommand{\pr}{\mathsf{price}}
 \newcommand{\Q}{\mathsf{Qty}}
 \newcommand{\q}{\mathsf{qty}}
 \newcommand{\id}{\mathsf{id}}

 \newcommand{\vol}{\mathsf{Vol}}

\newcommand{\select}{\mathsf{Select}}
\newcommand{\OPT}{\mathsf{OPT}}
\newcommand{\selectq}{\mathsf{SelectQ}}

\newcommand{\range}{\mathsf{Range}}
\newcommand{\spl}{\mathsf{Split}}
\newcommand{\splq}{\mathsf{SplitQ}}

\begin{document}
\title{The Exchange Problem}

\author{Mohit Garg\thanks{Indian Institute of Science, Bengaluru, India. mohitgarg@iisc.ac.in. Supported by a Walmart fellowship.} \ and
Suneel Sarswat\thanks{Tata Institute of Fundamental Research, Mumbai, India. suneel.sarswat@gmail.com.}}

\date{}

\maketitle\begin{abstract}
Auctions are widely used in exchanges to match buy and sell requests. Once the buyers and sellers place their requests, the exchange determines how these requests are to be matched. The two most popular objectives used while determining the matching are maximizing volume at a uniform price and maximizing volume with dynamic pricing. In this work, we study the algorithmic complexity of the problems arising from these matching tasks.

We present a linear time algorithm for uniform price matching which is an improvement over the previous algorithms that take $O(n\log n)$ time to match $n$ requests. For dynamic price matching, we establish a lower bound of $\Omega(n \log n)$ on the running time, thereby proving that the currently known best algorithm is time-optimal.
\end{abstract}
  
\paragraph{Key words:} Call Auctions,
Matching Algorithms, Element Distinctness,
Time Complexity.

\section{Introduction}
We study problems of the following nature: The input is a list of trade requests from buyers and sellers of a particular product. Each request consists of a price and a quantity. 
The buyer's request, known as a {\it bid}, is represented by a pair $(p,q)$, which indicates that the buyer offers to buy a maximum of $q$ units of the product while paying at most $p$ per unit. Similarly, a seller's request, known as an {\it ask}, is also a pair $(p,q)$, which indicates that the seller offers to sell a maximum of $q$ units of the product provided they receive at least $p$ per unit. 
On receiving the input, we are required to generate a {\it matching}, which is a collection of {\it transactions}.\footnote{We use the term {\it matching}, which has been traditionally used for this object in the auction theory literature~\cite{NiuP13,zhao2010maximal,WWW98,RSS21}, while a term like {\it flow} might be more suitable for a computer science audience.} A transaction between a bid $(p_b, q_b)$ and an ask $(p_a, q_a)$ consists of a transaction price $p\in [p_a, p_b]$ and a transaction quantity $q\leq\min\{q_b, q_a\}$, indicating that $q$ units of the product change hands between the traders at price $p$ per unit. For each trade request $(p,q)$, the sum of the total transaction quantities of the transactions involving the request must be at most $q$. The volume of the matching is the total quantity that changes hands. The matching is {\it uniform} if all transactions occur at a common price. We consider the following two tasks.

\begin{description}
\item[Task 1:] Determine a matching with the largest volume.

\item[Task 2:] Determine a uniform matching with the largest volume.
\end{description}

The first task can be solved by formulating it as a max-flow problem. However, due to the underlying problem structure, simpler solutions are available for these tasks.
In fact, two simple and almost identical algorithms can be used for these tasks. We describe them below. 

\label{twoAlgs}
\begin{description}
\item[Algorithms:] These algorithms start with sorting the list of bids in decreasing order of their prices. Next, the list of asks is sorted. For the first task, the sell requests are sorted in decreasing order of their prices. Whereas, for the second task, the sell requests are sorted in increasing order of their prices.

After the sorting step, both algorithms work in linear time as follows.  The bid on top of its sorted list $(p_b, q_b)$ is matched with the ask on top of its sorted list $(p_a, q_a)$ if they are compatible, i.e., $p_b \ge p_a$. In this case, a transaction between them is established with quantity $q= \min(q_b,q_a)$; a quantity of $q$ is reduced from their existing quantities of $q_b$ and $q_a$; finally, the $0$ quantity requests are deleted from the lists. If the requests are incompatible, the topmost ask in the sorted list of asks is deleted. 
The above steps are then repeatedly applied until one of the lists becomes empty. Finally, for the first task, for each matched pair: $\langle (p_b,q_b), (p_a, q_a) \rangle$ any transaction price in the interval $[p_a,p_b]$ can be put. For the second task, let the last matched pair in the above process be $\langle (p_b,q_b), (p_a, q_a) \rangle$. Then, for all matched pairs any number in the interval $[p_a,p_b]$ can be chosen as the common transaction price.
\end{description}

Clearly, the above algorithms take time $O(n\log n)$ in the comparison model. We ask the following question. 
\begin{description}
\item\textbf{Do these algorithms have the optimal running time for the above tasks?}
\end{description}

The problems we consider in this work are slight generalizations of the above tasks where bids and asks have additional parameters, and the objectives incorporate a notion of {\it fairness}. These problems 
arise out of {\it call auctions}, which are a class of {\it double auctions}, that are extensively used to match the buy and sell requests of traders at various centralized marketplaces such as exchanges of 
stocks
and commodities. Numerous works have focused on the mechanism-design and game-theoretic aspects of such auctions (e.g., see~\cite{ESA2002} and the citations therein), with only a few of them that include a runtime analysis of their algorithms~(\cite{WWW98,NiuP13,zhao2010maximal}). The problems considered in this work arise all the time in the real world and could easily be introduced in an undergraduate course in algorithms. For the above tasks, we prove the following.

\begin{theorem}\label{thm:result1Baby}
    Any algorithm for Task 1 requires $\Omega(n \log n)$ time in the worst-case, where $n$ is the number of trade requests.
\end{theorem}

\begin{theorem}\label{thm:result2Baby}
    There exists an algorithm for Task 2 that runs in $O(n)$ time, where $n$ is the number of trade requests.
\end{theorem}

Thus, interestingly, for Task 1, the algorithm described above is time-optimal, whereas for Task 2, the algorithm described above is not time-optimal.

To describe the general problems that we consider, previous results, and our contributions formally, we need the following preliminaries.
%%%%%%%%%%%%%%%
\subsection{Preliminaries}
 We now describe call auctions which are slight generalizations of the tasks introduced above. A call auction is used at an exchange to match the bids and asks of traders for a particular product. Bids and asks are collectively termed orders. 
Orders are collected by the exchange for a fixed duration of time, at the end of which the collected orders are simultaneously matched to produce transactions. Each transaction between a bid and an ask consists of a transaction price and a transaction quantity.
Each order consists of four attributes that are natural numbers: a unique identification number (id), a unique timestamp, a limit price, and a quantity that is at least one.\footnote{Later, the orders given as input to the algorithms, are all assumed to have distinct ids and distinct timestamps.} $\ts(w)$ represents the time the order $w$ was received. For an order $w$, the limit price $\pr(w)$ is the minimum (if $w$ is an ask) or maximum (if $w$ is a bid) possible transaction price of a transaction involving $w$. For an order $w$, the quantity $\q(w)$ represents the total quantity offered for trade, i.e., the sum of transaction quantities of transactions involving $w$ in the matching must be at most $\q(w)$. Given a collection of orders with distinct ids, i.e., given a set of bids and asks with distinct ids, a matching is a collection of transactions that satisfy the natural constraints. We formally define a matching below, along with some other important terms.

\begin{definition}[tradable, traded quantity $\Q$, volume $\vol$, matching] 

Given a bid $b$ and an ask $a$, we say that $b$ and $a$ are {\bf tradable} if $\pr(b)\geq \pr(a)$. 

For a set of transactions $M$ and an order $w$, $\Q(w,M)$ denotes the sum of transaction quantities of transactions in $M$ that involve $w$, and $\vol(M)$ denotes the sum of all transaction quantities of the transactions in $M$. With slight abuse of notation for a set of asks or a set of bids $\Omega$, we define $\vol(\Omega)$ to be the sum of the quantities of the orders in $\Omega$.

A set of transactions $M$ is a {\bf matching} over a list of bids $B$ and asks $A$ with distinct ids if (i) For each transaction $m\in M$, the bid and ask of $m$ come from $B$ and $A$, respectively. (ii) For each transaction $m \in M$, the bid $b$ of $m$ is tradable with the ask $a$ of $m$. (iii)
For each order $w \in B\cup A$,  $\Q(w, M) \leq \q(w)$. 
\end{definition}
There are two main types of call auctions: uniform-price and dynamic-price.
In both these auctions, the matching produced needs to be fair, which essentially states that more competitive orders (based on price-time priority) have to be given preference in the matching. We now formally define a fair matching along with the notion of competitiveness.
\begin{definition}[Competitiveness $\succ$, Fair] For bids $b_1$ and $b_2$ we say $b_1$ is more {\bf competitive} than $b_2$, denoted by $b_1 \succ$ $b_2$, iff $\pr(b_1) > \pr(b_2)$ or ($\pr(b_1)=\pr(b_2)$ and $\ts(b_1) < \ts(b_2)$). Similarly, for asks $a_1$ and $a_2$, we say $a_1$ is more competitive than $a_2$, denoted by $a_1 \succ a_2$, iff $\pr(a_1) < \pr(a_2)$ or ($\pr(a_1)=\pr(a_2)$ and $\ts(a_1) < \ts(a_2))$.

We say a matching $M$ is fair on bids if a bid $b$ participates in $M$, then all bids that are more competitive than $b$ must be fully traded in $M$. Formally, $M$ is fair on bids iff for all pairs of bids $b_1$ and $b_2$ such that $b_1 \succ b_2$, $\Q(b_2,M)\geq 1\implies \Q(b_1,M)=\q(b_1)$.
Similarly, a matching $M$ is fair on asks iff  for all pairs of asks $a_1$ and $a_2$ such that $a_1 \succ a_2$, $\Q(a_2,M)\geq 1\implies \Q(a_1,M)=\q(a_1)$. A matching is {\bf fair} if it is fair on bids as well as fair on asks. 
\end{definition}

In uniform price call auctions, the exchange is supposed to output a fair matching with maximum volume subject to the constraint that all transaction prices in the matching are identical. Customarily, the transaction price of a uniform price matching is referred to as the equilibrium price, and the process of discovering this price is called price discovery. In contrast, for dynamic price call auctions, the exchange is supposed to output a fair matching with maximum volume, where the transaction prices need not be the same. It must be noted that dropping the requirement of fairness from a dynamic price matching can only make the problem easier; the following result states that for any matching, there exists a fair matching of the same volume.

\begin{theorem}[Theorem 15 and Page 13 of \cite{RSS21}]\label{thm:fairness}
Given a set of bids $B$, a set of asks $A$, and a matching $M$ over $(B,A)$, one can find a fair matching $M'$ over $(B,A)$ such that 
\begin{itemize}
    \item[(a)] $\vol(M)=\vol(M')$, and
    \item[(b)] if $M$ is uniform, then $M'$ is uniform.
\end{itemize}
\end{theorem}

\begin{remark} It is clear from the above theorem, that
    unlike uniformity, ensuring fairness does not change the volume of a matching. More precisely, given sets of bids $B$ and asks $A$, consider the following matchings over $(B,A)$. Let $M_1$ be a largest volume matching, $M_2$ be a fair matching with the largest volume, $M_3$ be a uniform matching with the largest volume, and finally let $M_4$ be a uniform and fair matching with the largest volume. Then,  $$\vol(M_1)=\vol(M_2)\geq \vol(M_3)=\vol(M_4).$$
    Nevertheless, computing $M_1$ can be easier than computing $M_2$, and similarly computing $M_3$ can be easier than computing $M_4$. More precisely, computing $M_1$ reduces to computing $M_2$, and computing $M_3$ reduces to computing $M_4$.
\end{remark}

The two matching problems that we study in this work can be stated as follows.

\begin{description}

 \item[Problem 1] Given a set of bids $B$ and a set of asks $A$, find a fair matching over $(B,A)$ of maximum volume.
 \item[Problem 2] Given a set of bids $B$ and a set of asks $A$, find a fair and uniform matching over $(B,A)$ of maximum volume.
\end{description}

To show our lower bound results, we have to explain the model of computation we assume. 
\begin{description}
    \item[The comparison model:] We assume that the prices of the orders in the input are not given explicitly. 
Rather, two prices can be compared to each other via an oracle in unit time.

\item[The binary query model:] We will also derive a lower bound in a more general model, namely the binary query model, where the oracle can be made to evaluate an arbitrary function $f: \mathbb N \times \mathbb N \rightarrow \{0,1\}$ on two prices in unit time.
\end{description}
\begin{remark} Our upper bound results do adhere to the more restricted comparison model.
\end{remark}

We are now ready to state the previous and new results.

\subsection{Results, techniques, and discussion}

We first briefly summarize the past results that were obtained for Problems 1 and 2. In the following, $n$ represents the number of orders in the input, i.e., $n=|A|+|B|$.
\begin{itemize}
    \item First, in \cite{WWW98} it was shown that Problem 2 can be solved in $O(n \log n)$ time.
    \item Next, in \cite{zhao2010maximal} it was shown that Problem 1 can be solved in $O(n^2)$ time under the assumption that all orders have unit quantity. The authors mention that for multi-unit orders, one can simply break an order into multiple single-unit orders and use their algorithm. But, this would then result in overall complexity $O(Q^2)$, where $Q$ is the sum of the total demand and total supply, i.e., $Q=\vol(A)+\vol(B)$. 
    \item This was improved in \cite{NiuP13} where it was shown that Problem 1 can be solved in $O(n \log n)$ time for single-unit orders. 
    \item Finally, inspired from previous works, in~\cite{RSS21} algorithms for both Problems 1 and 2 were presented that run in $O(n \log n)$ time each for arbitrary multi-unit orders. However, the authors did not analyze the time complexity of their algorithms, as they were more focused on formalizing these algorithms in a theorem prover. Their algorithms are general versions of the two algorithms we saw on page~\pageref{twoAlgs}, where instead of sorting the lists of bids and asks by their prices, they are sorted by their competitiveness. After this modification, the first algorithm yields a maximum volume matching (note that Problem 1 requires the output matching to be also fair), whereas the second algorithm yields a uniform and fair matching with maximum volume (precisely solving Problem 2). For solving Problem 1, in addition to the above algorithm, they use an $O(n\log n)$ time algorithm that takes an arbitrary matching and outputs a fair matching of the same volume (this is always possible, see \cref{thm:fairness}). It is straightforward to see that their algorithms indeed run in $O(n\log n)$ time.
\end{itemize}

Now we describe the results obtained in this work.

For Problem 1, we show that an easier problem, namely finding a maximum volume matching (which need not be fair), requires $\Omega(n \log n)$ time for processing $n$ requests in the comparison model and $\Omega(n \sqrt{\log n})$ time in the binary query model.

\begin{restatable}[Result 1]{theorem}{resultone}
\label{thm:result1}
Any algorithm that takes as input a set of bids $B$ and a set of asks $A$ and computes a matching over $(B,A)$ with the maximum volume has a worst-case running time of $\Omega (n \log n)$ in the comparison model and $\Omega(n \sqrt{\log n})$ in the binary query model, where $n=|B|+|A|$.
\end{restatable}

From our earlier discussion, Theorem~\ref{thm:fairness} implies that these lower bounds also apply to Problem 1. 

As described above, there is an algorithm~\cite{NiuP13, RSS21} for finding a fair and maximum volume matching (which also works for determining a maximum volume matching) that takes $O(n \log n)$ time. Our lower bound result implies that such an algorithm is time-optimal in the comparison model. 

As we will see later, the Proof of \cref{thm:result1} also applies to \cref{thm:result1Baby}. In this proof,
we obtain the desired lower bound by reducing the element distinctness problem over a small domain to the maximum volume matching problem.\footnote{An entropy-based argument similar to the one used for sorting in~\cite{radhakrishnan2003entropy} also yields a lower bound for the maximum matching problem in the comparison model, but it fails to deliver any non-trivial lower bound in the binary query model.} Our result then follows from applying known lower bounds on the time complexity of the element distinctness problem~\cite{B94, M84}. We include a proof of the lower bound in the comparison model making the argument in~\cite{M84} more precise as it seems to be missing some important details. Our reduction uses a stronger version of the demand-supply inequality of~\cite{NiuP13}, which states that the volume of any matching is upper bounded by the sum of the total demand and the total supply at any price.

\begin{theorem}[demand-supply inequality]
    If $M$ is a matching over a set of bids $B$ and a set of asks $A$, then for all numbers $p$, we have
$$\vol(M) \leq \vol(B_{\geq p}) + \vol(A_{\leq p})$$
where $B_{\geq p}\subseteq B$ consists of bids whose limit prices are at least $p$, $A_{\leq p}\subseteq A$ consists of asks whose limit prices are at most $p$.
\end{theorem}

Our strengthening, which is crucially required for our reduction to work, comes from observing that the above inequality holds even if we replace  $\vol(B_{\geq p})$ by $\vol(B_{> p})$ or $\vol(A_{\leq p})$ by $\vol(A_{< p})$, and can be stated as follows.

\begin{restatable}[Result 2]{theorem}{resulttwo}
\label{thm:result2}
If $M$ is a matching over a set of bids $B$ and a set of asks $A$, then for all numbers $p$, we have
$$\vol(M) \leq \vol(B_{>p}) + \vol(A_{<p}) + \min(\vol(B_{=p}), \vol(A_{=p})),$$
where $B_{> p}\subseteq B$ consists of bids whose limit prices are greater than $p$, $A_{< p}\subseteq A$ consists of asks whose limit prices are less than $p$, and $B_{=p}\subseteq B$ (or $A_{=p}\subseteq A$) consists of bids (or asks) whose limit prices are exactly $p$.
\end{restatable}

For Problem 2, we obtain a linear time algorithm.

\begin{restatable}[Result 3]{theorem}{resultthree}
\label{thm:result3}
There exists an algorithm that takes as input a set of bids $B$ and a set of asks $A$, and outputs a uniform and fair matching with maximum volume in $O(n)$ time, where $n=|B|+|A|$.
\end{restatable}

As noted earlier, this is an improvement over the previous algorithms~\cite{WWW98,zhao2010maximal,RSS21} that take $O(n\log n)$ time to match $n$ requests. 
Note that \cref{thm:result3} immediately implies \cref{thm:result2Baby}, as ensuring fairness does not reduce the size of a uniform matching, as observed earlier.

Our improvement for uniform-price matching, roughly speaking, comes from getting rid of the sorting step in the algorithm of~\cite{RSS21} as described earlier. We instead compute the "medians" of the bids and asks in linear time which bisects the sets of bids and asks into two "equal halves" each. Based on whether the medians are tradable or not, we either can get rid of half of the bids and asks, or we can match half of the input in linear time, thus in either case reducing the problem size by half. This results in a linear time algorithm. Note that obtaining a linear time algorithm in the volume of the input orders is much simpler than our result.

\subsection{Other related work}

In \cite{GS22,lpar23}, continuous auctions were studied, that form a class of double auctions that complement the call auctions. The authors presented an algorithm that implements the continuous auction and runs in $O(n \log n)$ time for processing $n$ {\it instructions} (requests/orders). They also showed that their algorithm is time-optimal by reducing the task of sorting to continuous auctions.

\vspace{0.5cm}\noindent\textbf{Organization of the paper.} In \cref{sec:strongDemandSupply}, we prove \cref{thm:result2}. Then, \cref{thm:result1,thm:result1Baby} are proved in Section~$\ref{sec:maximum}$. Next, in \cref{sec:uniform}, \cref{thm:result3} is proved. Finally, we conclude the paper in \cref{sec:conclusions} with some concluding remarks.

\section{Stronger demand-supply inequality}\label{sec:strongDemandSupply}
In this section, we prove Theorem~\ref{thm:result2}.

\resulttwo*

\begin{proof}
[Proof of \cref{thm:result2}]
First observe that the volume of any matching is upper bounded by the volume of all bids as well as the volume of all asks, i.e., if $M$ is a matching over $(B,A)$, then 
$$\vol(M) 
\le \vol(B)\quad\text{ and } \quad\vol(M) \le \vol(A).$$
 To prove Theorem~\ref{thm:result2}, it suffices to prove the following two inequalities. 

\begin{align*} \vol(M) &\leq \vol(B_{> p}) + \vol(A_{< p}) + \vol(B_{=p});\\
 \vol(M) &\leq \vol(B_{> p}) + \vol(A_{< p}) + \vol(A_{=p}).
\end{align*}

To prove the first inequality we partition the matching $M$ into two sets: $M_1=\{(b,a,q,p') \in M \mid  \pr (b) \geq p\}$
consisting of all transactions in $M$ whose participating bid price is at least $p$ 
and 
$M_2=\{(b,a,q,p') \in M \mid  \pr (b) < p\}$ 
consisting of all transactions in $M$ whose participating bid price is strictly less than $p$. 
Thus, $\vol(M) = \vol(M_1) + \vol(M_2)$. 

It is easy to see that $M_1$ is a matching over sets of bids $B_{\geq p}$ and asks $A$, and hence from the above observation, $\vol(M_1) \leq \vol(B_{\geq p}) = \vol(B_{> p}) + \vol(B_{= p})$.

Next, we prove that $M_2$ is a matching over sets of bids $B$ and asks  $A_{< p}$. Consider a transaction $m=(b,a,q,p')$ from $M_2$, which is between a bid $b$ and an ask $a$. Since $m \in M$, $\pr (b) \ge \pr (a)$, and from the definition of $M_2$, we have $\pr (b) < p$. This implies $\pr(a) < p$, i.e., asks of $M_2$ come from $A_{< p}$. Hence, $M_2$ is a matching over $(B,A_{< p})$, and applying the above observation again, we have $\vol(M_2) \le \vol(A_{< p})$.

Combining, we have $\vol(M) = \vol(M_1) + \vol(M_2) \leq \vol(B_{> p}) + \vol(B_{= p}) + \vol(A_{< p})$, which completes the proof of first inequality. 

Similarly, we can prove the second inequality by first partitioning the matching $M$ into $M_1=\{(b,a,q,p') \in M \mid  \pr (b) > p\}$ and $M_2=\{(b,a,q,p') \in M \mid  \pr(b) \leq p\}$, noticing that $M_1$ is a matching over $(B_{>p}, A)$ and $M_2$ is matching over $(B, A_{ \leq p})$, and finally, using the observation again to obtain the desired inequality. 
\end{proof}

\section{Maximum volume matching lower bound}
\label{sec:maximum}
In this section we prove Theorem~\ref{thm:result1} (note the proof also works for \cref{thm:result1Baby}), i.e., we establish a lower bound of $\Omega(n \log n)$ on the running time of any algorithm that computes a maximum volume matching on a list of bids $B$ and a list of asks $A$, where $n=|B|
+|A|$. Let $\MM$ be an algorithm for the maximum matching problem. We reduce the element distinctness problem on a small domain to $\MM$ in linear time. 

\noindent\textbf{Element Distinctness Problem (on small domain)}: Given an input $(x_1,\cdots, x_n)$, where each $x_i\in[n]$, check whether there are distinct indices $i$ and $j$ such $x_i = x_j$.\footnote{For $n \ge 1$, $[n]$ denotes the set $\{1, 2, \ldots , n\}$.}

Our result then immediately follows from the fact that the element distinctness problem on a small domain requires time $\Omega(n \log n)$ in the comparison model~\cite{M84} (for which we include a proof in the next subsection) and time $\Omega(n\sqrt{\log n})$ in the binary query model~\cite{B94}.

In the reduction, we will make use of the stronger version of the demand-supply inequality Theorem~\ref{thm:result2}, which was proved in~\cref{sec:strongDemandSupply}.

\noindent\textbf{The reduction}.
Given an instance $X = (x_1,\cdots, x_n)$ for the element distinctness problem, we construct two sets of orders $\Omega = \{\omega_1, \omega_2, \ldots, \omega_n\}$ and $\Lambda = \{\lambda_1, \lambda_2, \ldots, \lambda_n\}$ such that the quantity of each order in $\Omega \cup \Lambda$ is set to $1$. We set the prices as follows. $\pr(\omega_i) = i$ and $\pr(\lambda_i) = x_i$. Observe that the prices of orders in $\Omega$ are all distinct. Next, we run the maximum matching algorithm $\mathsf{MM}$ twice on these inputs by first treating $\Omega$ as the set of bids and then treating $\Omega$ as the set of asks to obtain two matchings $M_1$ and $M_2$, respectively; $M_1 = \mathsf{MM}(\Omega, \Lambda)$ and $M_2 = \mathsf{MM}(\Lambda, \Omega)$. 

We now claim that if elements of $X$ are all distinct then $\vol(M_1) = \vol(M_2) = n$, otherwise $\vol(M_1) < n \text{ or } \vol(M_2) < n$. Note that if we show this claim the reduction is complete, as we can then solve the element distinctness problem mentioned above in time $O(n)$ plus twice the time taken by $\mathsf{MM}$ on $2n$ orders.

It is easy to see that if the $x_i$'s are all distinct, then  $\vol(M_1)=\vol(M_2)=n$; the bid with price $i$ is matched with the ask with price $i$.

Now, we show that if the $x_i$'s are not distinct, then $\vol(M_1) \leq n-1$ or $\vol(M_2)\leq n-1$. From the pigeon-hole principle, if two elements are repeating in $X$, then one of the elements from $[n]$ must be missing from $X$. Let the smallest missing element in $X$ be $m$ and the smallest repeating element in $X$ be $r$. 

For a set of orders $\Omega$, we define $\Omega_{=t}:=\{\omega\in\Omega \mid \pr(\omega)=t\}$,  $\Omega_{<t}:=\{\omega\in\Omega \mid \pr(\omega)<t\}$, and $\Omega_{>t}:=\{\omega\in\Omega \mid \pr(\omega)>t\}$.
Now, observe that $\vol({\Lambda}_{=m}) = 0$ (since $m$ is missing from $X$ which is the set of prices of $\Lambda$), which implies $\min(\vol({\Omega}_{=m}), \vol({\Lambda}_{=m})) = 0$.

We now consider two cases: $m < r$ or $m > r$.

\noindent\textbf{Case: $m < r$.} In this case we will show that $\vol(M_1) < n$. Observe that the number of elements in $X$ that are at most $m$ is precisely $m-1$, since no element smaller than $m$ is repeating and $m$ itself is missing. Here, the input $x_i$'s are set to be the ask prices, i.e., $\Lambda$ is the set of asks, and $\Omega$ is the set of bids. Therefore, $\vol({\Lambda}_{< m}) = |\{\lambda \in \Lambda \mid \pr(\lambda) \le m\}| = m-1$ and $\vol({\Omega}_{>m}) = n - m$ . Hence, from Theorem~\ref{thm:result2}, \begin{align*}
    \vol(M_1) &\le \vol({\Omega}_{> m}) + \vol({\Lambda}_{< m}) + \min(\vol({\Omega}_{=m}),\vol({\Lambda}_{=m}))\\ &=  (n - m) + (m - 1) + 0 = n - 1.
\end{align*}

\noindent\textbf{Case: $m > r$. } In this case we will show that $|M_2| < n$. Note that no element in $X$ is missing below $m$ but some elements are repeating. Thus, the number of elements that are greater than $m$ are at most $n-m$. Hence, similar to the above case, applying Theorem~\ref{thm:result2}, we have $
    \vol(M_2) \le \vol({\Lambda}_{> m}) + \vol({\Omega}_{< m}) + \min(\vol({\Lambda}_{=m}),\vol({\Omega}_{=m}))\leq  (n - m) + (m - 1) + 0 = n - 1.$\qed

\subsection{Element distinctness on small domain}

We now show a lower bound on the number of comparisons needed to solve the element distinctness problem on a small domain: Given an input $(x_1,\cdots, x_n)$ where each $x_i\in[n]$, check whether there exists $i$ and $j$ such that $i\neq j$ and $x_i=x_j$.

Any algorithm for element distinctness can be seen as a decision tree where each internal node is labeled by a comparison $X_i : X_j$, for some $i\neq j$, and each leaf is labeled either Yes or No. On a given input $(x_1,\cdots, x_n)$ the decision tree is applied as follows. We start at the root node, traverse a path to a leaf, and declare the label of the leaf as the output. On encountering a node labeled $X_i:X_j$, if $x_i\leq x_j$, we take its left child as the next node in the path, otherwise, we take its right child. Inputs corresponding to permutations on $[n]$ must receive the answer Yes, and inputs where $x_i=x_j$ for some $i\neq j$ must receive the answer No. 

We fix a decision tree $T$ for the element distinctness problem and prove that its height is at least $\Omega (n \log n)$. We make the following claim.

\begin{claim}
    Each permutation on $[n]$ ends up in a distinct leaf of $T$. 
\end{claim}

This immediately implies that the number of leaves is at least $n!$ and hence the height of the tree is at least $\Omega(\log n!)= \Omega(n \log n)$. We now turn to proving the above claim.

For the sake of contradiction we assume two distinct permutations $\sigma=(\sigma_1,\cdots, \sigma_n)$ and $\tau=(\tau_1,\cdots,\tau_n)$ end up at the same leaf $L$ of $T$. To obtain a contradiction we will find an input $(S_1,\cdots, S_n)$ where $S_i= S_j$ for some $i\neq j$ which also ends up in the leaf $L$.

We define a poset $P(L)$ on the set of symbols $X=\{X_1,\cdots,X_n\}$. In our poset when $X_i$ and $X_j$ are related, we will write $X_i \leq X_j$ and call it a constraint. For each node labeled $X_i: X_j$ in the path from the root to the leaf $L$ in the tree $T$, we have a constraint of the from $X_i \leq X_j$ or $X_j \leq X_i$ (in fact, $X_j < X_i$ suffices, but we use the weaker form) which must be satisfied for a computation to take this path. Let $C$ be the set of all such constraints. We define $C(L)= \bigcap \{\hat C \mid C\subseteq \hat C \text{ such that $(X,\hat C)$ is a poset}\}$; notice that this is a non-empty intersection, as we have at least one poset, namely the total order $X_{\sigma(1)} \leq X_{\sigma(2)} \leq \cdots \leq X_{\sigma(n)}$ which satisfies all constraints in $C$ as $\sigma$ ends up in the leaf $L$. We define $P(L):= (X, C(L))$. 

We say an input $(x_1,\cdots, x_n)$ {\it respects} a poset $P$ if $x_i \leq x_j$ whenever $X_i \leq X_j$ is a constraint in $P$. The following proposition is easy to prove.

\begin{proposition}
An input reaches the leaf $L$ iff it respects the poset $P(L)$.
\end{proposition}

Next, observe that since $\sigma$ and $\tau$ are two distinct permutations that respect $P(L)$, the length of a largest chain in $P(L)$ is at most $n-1$.\footnote{At this point, the proof in~\cite{M84} instantly concludes that there must be two elements in the poset that are not related and fixes an input which respects the poset and where these two coordinates are equal, to complete the proof, without specifying why such an input exists. We fill such gaps in the proof.} Mirsky's theorem~\cite{Mirsky} then implies that there is an antichain decomposition of $P(L)$ with at most $n-1$ antichains.
Since $X$ has $n$ elements, there must be an antichain in $P(L)$ that has at least two elements. Let $A$ be a maximal antichain in $P(L)$ that has at least two elements. We can write $X$ as a disjoint union of three sets $X =  Q \uplus A \uplus R $, where $
Q$ consists of elements below $A$ and $R$ consists of elements above $A$ in the poset $P(L)$: $ Q= \{ X_i \in X \setminus A \mid \text{ there exists } X_j\in A \text{ such that } X_i \leq X_j \}$ and $R = \{ X_i \in X \setminus A \mid \text{ there exists } X_j\in A \text { such that } X_j \leq X_i \}$ (observe that an element cannot be both above and below $A$). The constraints in $C(L)$ restricted to $Q$ and restricted to $R$ give rise to the posets: $P_Q$ and $P_R$. Now we define the poset $(X,D)$ such that $C(L)\subseteq D$ as follows. Take a linear extension of $P_Q$: $X_{t_1} \leq \cdots \leq X_{t_q}$ and add all the implied constraints to $D$. Do not add any constraints between the elements of the antichain $A=\{X_{t_{q+1}}, X_{t_{q+2}},\cdots, X_{t_{q+a}}\}$. Take a linear extension of $P_R$:
$X_{t_{q+a+1}}\leq \cdots \leq X_{t_n}$ and add all the implied constraints to $D$. Furthermore, add the following constraints to $D$: For each $X_i\in X$, add $X_i\leq X_i$.
For each $X_i\in Q$ and $X_j\in A\cup R$, add $X_i\leq X_j$. For each $X_i\in A$ and each $X_j\in R$, add $X_i\leq X_j$. Clearly, $C(L)\subseteq D$.
Now fix the input $S$ such that
$S_{t_1} = 1,S_{t_2} = 2,\cdots, S_{t_q}=q
$, $S_{t_{q+1}}=S_{t_{q+2}}=\cdots= 
S_{t_{q+a}}= q+1$,
$S_{t_{q+a+1}}=q+2, S_{t_{q+a+2}}=q+3,\cdots, S_{t_n} = n-(a-1))$.
Now, $S$ clearly respects the poset $(X,D)$. Thus, it also respects the poset $P(L)$  (as $C(L)\subseteq D$) and reaches the leaf $L$. Since $a =|A|\geq 2$, $S$ is a No instance which reaches $L$, a contradiction.\qed

\section{Uniform price matching}
\label{sec:uniform}
In this section, we prove Theorem~\ref{thm:result3}, i.e., we obtain a linear time algorithm for uniform price matching. We begin by first describing the previous $O(n \log n)$ algorithm as described in~\cite{RSS21}, which we will critically use in establishing the correctness of our linear time algorithm.

We may assume that the input lists of asks $A$ and bids $B$ are such that $\vol(A) = \vol(B)$. This can be achieved by adding a dummy order. Say $\vol(A) < \vol (B)$, then add a dummy ask with $\pr = \infty$ and $\q = \vol(B) - \vol (A)$. If $\vol(A) > \vol(B)$, then add a dummy bid with $\pr = -1$ and $\q = \vol(A) - \vol(B)$. Strictly speaking, according to our definition, we cannot set the $\pr$ to be $\infty$ or $-1$, but our definitions can be adjusted to allow for this.

\subsection{Previous algorithm}
We describe the algorithm in~\cite{RSS21}, which we denote by $\UM$. The transaction prices that are set by $\UM$ are not guaranteed to be uniform initially. After the matching is produced, a linear time algorithm is employed to put a uniform transaction price for all the transactions. For example, the maximum of the limit prices of the asks that participate in the output matching can be set as the transaction price of all the transactions in the output matching.

Given the list of bids B and asks A, $\UM$ first sorts the lists based on their respective competitiveness with the most competitive order on the top. It then invokes $\match$ on the sorted lists $B$, $A$, and an empty matching $M$. Throughout the algorithm, $M$ can only grow, and at the end of the algorithm, $M$ will contain a desired uniform-price matching.

\begin{algorithm}[H]
 \caption{Uniform Matching $\UM$}\label{process1}
\begin{algorithmic} %To add line numbering change this to \begin{algorithmic [1]
\Function{$\UM$}{Bids $B$, Asks $A$}
\State Sort the lists $B$  and $A$ based on their respective competitiveness
\State \Return $\match(B,A,\emptyset)$
\EndFunction
\end{algorithmic} 
\end{algorithm}

$\match$ on $(B,A,M)$ picks the most competitive bid $b$ and the most competitive ask $a$ from $B$ and $A$ respectively, and checks if they are tradable. If they are not tradable, it returns $M$, and the algorithm terminates. 

Otherwise, if they are tradable, it adds a transaction between $b$ and $a$ with transaction quantity $q = \min\{\q(a),\q(b)\}$ to $M$, reduces the quantity of $a$ and $b$ by $q$ in the lists $B$ and $A$. Note that at least one of $b$ or $a$ must be fully exhausted and removed completely from $B$ or $A$. It then recursively calls $\match$ on $(B, A, M)$.

\begin{algorithm}[H]
\label{algo:match}
\caption{The $\match$ subroutine }\label{process2}
\begin{algorithmic} %To add line numbering change this to \begin{algorithmic [1]
 \Function{$\match$}{Bids $B$, Asks $A$, Matching $M$}
 \Comment{Initially, $M=\emptyset$.}

\If {$|B|=0$ or $|A|=0$}
\State \Return{M}
\EndIf
 \State
\State $b\leftarrow \mathsf{pop}(B)$
\State $a\leftarrow \mathsf{pop}(A)$
\State

\If {$\pr(b)<\pr(a)$}
\State $\mathsf{push}(B, b)$
\State \Return $\match(B, A, M)$
\EndIf
\Comment{otherwise, b and a are tradable}

\State $q\leftarrow \min\{\q(a),\q(b)\}$
\State
\State $\mathsf{push}(M, \{(\id(b),\id(a), q, \pr(a))\})$
\State
\If {$\q(b) - q > 0$}
\State $\mathsf{push}(B,(\id(b),\ts(b),\q(b)-q,\pr(b)))$
\EndIf

\If {$\q(a) - q > 0$}
\State $\mathsf{push}(A,(\id(a),\ts(a),\q(a)-q,\pr(a)))$
\EndIf

\State
\State \Return $\match(B,A,M)$
 \EndFunction    
 \end{algorithmic} 
 \end{algorithm}
 
Clearly $\UM$ takes $O(n \log n)$ time,  whereas $\match$ takes $O(n)$ time, where $n=|A|+|B|$.

The algorithm can be simply described as first sorting the lists $B$ and $A$ based on competitiveness, and then doing a top-down greedy matching as long as the current bid and ask are tradable. 

The following result states that the above result is correct.

\begin{theorem}[Theorem 20 of~\cite{RSS21}] $\UM(B,A)$ outputs a uniform price matching.
\label{thm:UMcorrect}
\end{theorem}

We give a brief intuition of the proof of the above result. Let $M= \UM(B,A)$. To see that $M$ is a fair matching is trivial, as it starts matching the bids and asks in the decreasing order of competitiveness. Similarly, it is easy to convince oneself that $M$ has a uniform price; any price between the limit prices of the last paired bid and ask is acceptable to all orders that participate in the matching; this again follows from the fact that matching is done top down in the order of decreasing competitiveness. Finally, to see that $M$ has maximum volume needs some work: fix an optimal matching $\OPT$. We can gradually transform $\OPT$ into $M$ without altering its volume. The illuminating case to consider is when the most competitive orders $b$ and $a$ are fully traded in $M$. Note that the transaction quantity between $b$ and $a$ in $M$ is $q= \min\{\q(b), \q(a)\}$. Since $\OPT$ is fair, $b$ and $a$ must also be fully traded in $\OPT$.  In particular $\Q(b,\OPT) \geq q$ and $\Q(a,\OPT) \geq q$.  Now, if the transaction quantity between $b$ and $a$ in $\OPT$ is strictly less than $q$, then $\OPT$ can be modified by making the transaction quantity between $b$ and $a$ equal to $q$. Let us say $b$ is matched with $a'$ and $a$ is matched with $b'$ in $\OPT$, then we can reduce the quantity of these transactions by a unit quantity, increase the transaction quantity between $b$ and $a$ by unit quantity, and increase the transaction quantity between $b'$ and $a'$ by a unit quantity (this is possible since every matched bid is tradable with every matched ask as all the transaction prices are identical). Note that in the end, $\vol(\OPT)$ remains the same. We can repeat this surgery over and over again till the matched quantity between $b$ and $a$ in $\OPT$ becomes equal to $q$. We remove this transaction from both $\OPT$ and $M$ and apply the same argument repeatedly till $M = \OPT$.

\subsection{Useful lemmas and subroutines}

Before we proceed to describe our improved algorithm, we establish certain lemmas and subroutines that will be useful in describing our algorithm and proving its correctness.

We first need some definitions.

\begin{definition}[$\spl$, $\range$]
For a set of bids or a set of asks $\Omega$ and an element $\omega\in\Omega$, $\spl(\Omega, \omega)$ returns a partition of $\Omega = (\Omega_\succeq, \Omega_\prec)$, where $\Omega_\succeq =\{x\in \Omega \mid x\succeq \omega\}$ and $\Omega_\prec=\{x\in \Omega \mid x \prec \omega\}$. Thus, $\spl$ splits $\Omega$ into two parts, one containing orders that are at least as competitive as $\omega$, and the other containing orders that are less competitive than $\omega$. Clearly, $\spl$ can be implemented in linear time.

We now define the range of an order in a set of orders $\Omega$ which are all bids or all asks.
Let $\omega$ be an order in $\Omega$. Let $\spl(\Omega,\omega) = (\Omega_\succeq, \Omega_\prec)$. Let $\Omega_\succ = \Omega_\succeq \setminus \{\omega\}$. $\range_\Omega(\omega) = \{x\in\mathbb{N} \mid \vol(\Omega_\succ) < x \leq \vol(\Omega_\succeq)\}$. If all orders in $\Omega$ have unit quantities, then the $\range_\Omega$ of the $i^\text{th}$ most competitive order is the singleton $\{i\}$. In general, range of the most competitive order $\omega_1$ is the set $\{1,\cdots,\q(\omega_1)\}$. The range of the next most competitive order $\omega_2$ is $\{\q(\omega_1)+1,\cdots, \q(\omega_1) + \q(\omega_2)\}$, and so on. Thus, the ranges of orders in $\Omega$ partition $[\vol(\Omega)]$.
\end{definition}

Observe the following fact about the previous algorithm $\UM$.
\begin{lemma}
    If a bid $b\in B$ and an ask $a\in A$ are matched in the matching output by $\UM(B,A)$, then $\range_B(b)\cap\range_A(a) \neq \emptyset$.
\label{lem:1}
\end{lemma}

This is trivial to see when all orders in $B\cup A$ have unit quantities since $\UM$ matches the $i^\text{th}$ most competitive bid with the $i^\text{th}$ most competitive ask for all $i \le \vol(M)$, where $M = \UM(B,A)$, and the $\range$ of the $i^\text{th}$ most competitive bid (and ask) is the singleton $\{i\}$. To see why the general statement is true observe that whenever $\match(\hat B,\hat A,\hat M)$ is called it potentially matches the most competitive bid $b \in \hat B$ with the most competitive ask $a \in \hat A$, and $\vol(\hat M) + 1 \in \range_B(b) \cap \range_A(a)$.

Our main workhorse is the $\select$ subroutine which works in linear time by employing the classical algorithm of~\cite{BFPRT73}. Given a list of orders (all bids or all asks) $\Omega$ and a number $t \leq |\Omega|$, $\select(\Omega, t)$ outputs the $t^{\text{th}}$ most competitive order in $\Omega$.

We also use a weighted version of $\select$ called $\selectq$ which takes as input $\Omega$ and a quantity $q \leq \vol(\Omega)$ and outputs the unique element  $\omega \in \Omega$ such that $q\in \range_\Omega(\omega)$.
$\selectq$ can be implemented in linear time using the $\select$ algorithm as subroutine as described below.

\begin{algorithm}[H]
 \caption{Weighted Selection where weights are quantities}\label{process3}
\begin{algorithmic}  %To add line numbering change this to \begin{algorithmic [1]
\Function{$\mathsf{SelectQ}$}{Orders $\Omega$, $q$}
\State $\omega= \mathsf{Select}(\Omega, \lceil \frac {|\Omega|}{2} \rceil)$ 
\State $(\Omega_\succeq,\Omega_\prec) \leftarrow \mathsf{Split}(\Omega, \omega)$
\State 
\If {$\vol(\Omega_\succeq) - \q(\omega) < q \le \vol(\Omega_\succeq)$}
 \State \Return $\omega$ 
 \EndIf

 \If {$q \leq \vol(\Omega_\succeq) - \q(\omega) $}
    \State \Return $\mathsf{SelectQ}(\Omega_\succeq, q)$
    \EndIf
 \If {$q > \vol(\Omega_\succeq) $}
\State \Return$\mathsf{SelectQ}(\Omega_\prec,q - \vol(\Omega_\succeq))$
\EndIf
 \EndFunction    
 \end{algorithmic} 
 
\end{algorithm}

We get the following recurrence relationship on the running time of $\selectq$: $T(n) \leq T(\lceil n/2 \rceil + 1) + O(n)$, which yields $T(n) = O(n)$, where $n = |\Omega|$.

The next subroutine we consider is $\spl(\Omega,\omega)$, which partitions $\Omega$ into two parts consisting of orders in $\Omega$ which are at least as competitive as $\omega$, and orders which are strictly less competitive than $\omega$. We now want to define a subroutine $\splq$ which takes as input a list of orders $\Omega$ and a quantity $q\leq \vol(\Omega)$. We want to `partition' $\Omega$ into two parts so that the volume of the more competitive part is precisely $q$. 
For this, we first find the element $\omega\in\Omega$ such that $q\in\range_\Omega(\omega)$. We then $\spl(\Omega,\omega)$ to obtain $(\Omega_\succeq,\Omega_\prec)$. If $\vol(\Omega_\succeq) = q$, then we are immediately done. Else $\vol(\Omega_\succeq) > q$ and $\vol(\Omega_\succeq \setminus \{\omega\}) < q$. Thus we must break a part of $\omega$ and put it in $\Omega_\prec$ so that $\vol(\Omega_\succeq)$ is precisely $q$. This can be achieved in linear time by the following subroutine $\splq(\Omega,q)$. Apart from outputting the partition as described above, it also returns the order $\omega$.  $\splq$ clearly runs in linear time.

\begin{algorithm}[H]
\caption{Splitting an order list by a particular quantity.}\label{process4}
\begin{algorithmic} %To add line numbering change this to \begin{algorithmic [1]
\Function{$\mathsf{SplitQ}$}{Orders $\Omega$, quantity $q$}\Comment{Promise: $q\le \vol(\Omega)$}
\State $w= \mathsf{SelectQ}(\Omega, q)$
\State  $(\Omega_\succeq,\Omega_\prec) \leftarrow  \mathsf{Split}(\Omega, \omega)$
\State
\State $q_\text{extra} = \vol(\Omega_\succeq) - q$
\State
\If {$q_\text{extra} > 0$}
\State $\Omega_\succeq \leftarrow [\Omega_\succeq \setminus\{\omega\} ]\cup \{(\id(\omega),\ts(\omega),\q(\omega)-q_\text{extra},\pr(\omega))\}$
\State $\Omega_\preceq \leftarrow \Omega_\prec \cup \{(\id(\omega),\ts(\omega),q_\text{extra},\pr(\omega))\}$
\EndIf
\State \Return($\omega$, $\Omega_\succeq$, $\Omega_\preceq$)
 \EndFunction    
 \end{algorithmic} 
 
\end{algorithm}

We are now ready to state another lemma regarding the previous algorithm $\UM$.

\begin{lemma}
    Let $\hat b\in B$ be the $t^\text{th}$ most competitive bid that gets completely traded in $M=\UM(B,A)$. Let $\spl(B, \hat b) = (B_\succeq, B_\prec)$ and $q = \vol (B_\succeq)$. Let $\splq(A, q) = (\hat a, A_\succeq, A_\preceq)$. Then, $M$ can be partitioned into $(M_1, M_2)$ such that $M_1 = \UM(B_\succeq, A_\succeq)$ and $M_2= \UM(B_\prec, A_\preceq)$.
    \label{lem:2}
\end{lemma}

The above statement is easy to check. Since $\hat b$ gets completely traded in $M$, all orders in $B_\succeq$ must get traded with the most competitive asks in $A$ whose total quantity is $q$, which is precisely the set $A_\succeq$. Let $M_1$ consist of the transactions produced while matching orders in $B_\succeq$. Then $M_2$ will be obtained by running the $\UM$ on $(B_\prec,A_\succeq)$.

Note that a symmetric statement holds where we start with the assumption that the $t^\text{th}$ most competitive ask gets traded in $M$.

We are now ready to describe our improved algorithm.
%%%%%%%%%%%%%%%%%%%%%%%%%%%%%%%%%%%%%%%%%%%%%%%%
\subsection{Improved algorithm}

Our improved algorithm $\UMs$ takes as input a list of bids $B$ and a list of asks $A$. It immediately invokes $\UMb$ on $(B,A,M=\emptyset)$. $M$ will grow into the final matching output by the algorithm.

\begin{algorithm}[H]
 \caption{Efficient Uniform Algorithm }\label{process5}
\begin{algorithmic} 
\Function{$\UMs$}{Bids $B$, Asks $A$}

\State $\UMb(B, A, \emptyset)$
 \EndFunction    
 \end{algorithmic} 
 
\end{algorithm}

There are two main subroutines of our algorithm $\UMb$ and $\UMa$ which are symmetric in nature and they alternatively call each other. This is done to ensure that in two successive return calls the problem size, i.e., $|B|+|A|$ reduces by a factor of two, as $\UMb$ halves the number of the bids, whereas $\UMa$ halves the number of the asks. So just understanding $\UMb$ will be sufficient for understanding our algorithm.

On receiving $(B,A,M)$ $\UMb$ first finds the median bid $b$ by invoking $\mathsf{Select}(B, \lceil\frac{B}{2} \rceil)$. It then splits $B$ into two halves $(B_\succeq, B_\prec)$ based on the median bid $b$ by invoking $\spl(B,b)$. Let $q = \vol(B_\succeq)$. It then finds the element $a\in A$ such that $q\in \range_A(a)$ and splits $A$ into its most competitive and least competitive asks $(A_\succeq, A_\preceq)$ such that $\vol(A_\succeq) = q$ by applying $\splq(A,q)$.

After that, the algorithm checks if $b$ and $a$ are tradable. Two cases arise.
If $b$ and $a$ are tradable, then every order in $B_\succeq$ and $A_\succeq$ are tradable and they will exhaustively be matched to each other to produce a matching of volume $q$ in linear time using the $\match$ subroutine. $\UMb$ adds this $\vol$ $q$ matching to $M$ and recursively calls $\UMa$ on $(B_\prec, A_\preceq, M)$.
Intuitively, the previous algorithm would also match all orders in $B_\succeq$ and $A_\succeq$ exhaustively to each other to produce a matching of size $q$ and proceed to matching orders in $B_\prec$ and $A_\preceq$.

In the case $b$ and $a$ are not tradable, then $\UMb$ discards $(B_\prec, A_\preceq)$ and calls $\UMa$ on $(B_\succeq, A_\succeq, M)$. Intuitively, the previous algorithm will halt (i.e., produce its last transaction) before it even comes down to examining orders in $B_\succeq$ and $A_\succeq$.

 \begin{algorithm}[H]
 \caption{Uniform Matching by bisecting  bids}\label{process6}
\begin{algorithmic}
\Function{$\UMb$}{Bids $B$, Asks $A$, Matching $M$}

\If {$|B| = 0$ or $|A| = 0$ or ($B = \{b\}$ and $A=\{a\}$ and $\pr(a) > \pr(b)$)}
\State\Return{$M$}
\EndIf
\State
\State $b= \mathsf{Select}(B, \lceil \frac {|B|}{2} \rceil)$

\State $(B_\succeq,B_\prec) \leftarrow \mathsf{Split}(B,b)$
\State

\State $(a,A_\succeq,A_\preceq) 
\leftarrow 
\mathsf{SplitQ}(A, \vol(B_\succeq))$
\State

\If {$\pr(a) \leq \pr(b)$}
\State 
$M \leftarrow M \cup \mathsf{Match}(B_\succeq,A_\succeq)$
\State \Return $\UMa(B_\prec,A_\preceq, M)$
\EndIf
\Comment{Otherwise, $b$ and $a$ are not tradable}
\State\Return $\UMa(B_\succeq,A_\succeq, M)$
 \EndFunction    
 \end{algorithmic} 
\end{algorithm}

 \begin{algorithm}[H]
 \caption{Uniform Matching by bisecting asks}\label{process7}
\begin{algorithmic}
\Function{$\UMa$}{Bids $B$, Asks $A$, Matching $M$}

\If {$|B| = 0$ or $|A| = 0$ or ($B = \{b\}$ and $A=\{a\}$ and $\pr(a) > \pr(b)$)}
\State\Return{$M$}
\EndIf
\State
\State $a= \mathsf{Select}(A, \lceil \frac {|A|}{2} \rceil)$

\State $(A_\succeq, A_\prec) \leftarrow \mathsf{Split}(A,a)$
\State

\State $(b, B_\succeq,B_\preceq) 
\leftarrow 
\mathsf{SplitQ}(B, \vol(A_\succeq))$
\State

\If {$\pr(a) \leq \pr(b)$}
\State 
$M \leftarrow M \cup \mathsf{Match}(B_\succeq,A_\succeq)$
\State \Return $\UMb(B_\preceq,A_\prec, M)$
\EndIf
\Comment{Otherwise, $b$ and $a$ are not tradable}
\State\Return $\UMb(B_\succeq,A_\succeq, M)$
 \EndFunction    
 \end{algorithmic} 
\end{algorithm}

Having described our algorithm $\UMs$, we now turn to prove its correctness. The main theorem that establishes the correctness is as follows.

\begin{theorem}\label{thm:UMiscorrect}
    Given a list of bids $B$ and a list of asks $A$, let $\OPT = \UM(B,A)$ and $M = \UMs(B,A)$. Then, for each order $w\in B\cup A$, $\Q(w, \OPT) = \Q(w, M)$.
\end{theorem}

Once we prove the above theorem, the correctness follows from the following proposition and Theorem~\ref{thm:UMcorrect}.

\begin{proposition}
    If $M_1$ is a uniform price matching over $(B,A)$ and $M_2$ is a matching over $(B,A)$ such that for all $w\in B\cup A$, $\Q(w,M_1) = \Q(w, M_2)$, then $M_2$ is a uniform price matching.
\end{proposition}

The proposition is obvious: the volumes of $M_1$ and $M_2$ must be the same from the condition above. Also, since the same orders participate in both $M_1$ and $M_2$, transactions in $M_2$ can be assigned the same uniform price that is in $M_1$. Finally, fairness also follows immediately since the more competitive orders are fully traded in $M_1$, they must be fully traded in $M_2$ as $\Q(w,M_1)=\Q(w,M_2)$ for all orders $w$.

We now turn to proving Theorem~\ref{thm:UMiscorrect}.
\begin{proof}[Proof of Theorem~\ref{thm:UMiscorrect}]
For a list of orders $\Omega$, we use $\Omega{\downarrow}$ to denote the list obtained by sorting $\Omega$ by decreasing competitiveness.  

We make the following claim.
\begin{claim}
     Let $B$ be a list of bids, $A$ be a list of asks, and $M$ be a matching such that $\vol(B)=\vol(A)$. Let $M_1 = \match(B{\downarrow }, A{\downarrow}, M)$, $M_2= \UMb(B,A, M)$, and $M_3=\UMa(B,A,M)$. Then, for all $\omega\in B\cup A$, $\Q(\omega, M_1) = \Q(\omega, M_2) = \Q(\omega, M_3)$.
\end{claim}

If the claim is true, then Theorem~\ref{thm:UMiscorrect} follows immediately by observing that without loss of generality we had assumed that our inputs $B$ and $A$ are such that $\vol(B)=\vol(A)$ (which was achieved by adding a dummy untradable order), $\UM(B,A)=\match(B{\downarrow},A{\downarrow},\emptyset)$, and $\UMs(B,A)=\UMb(B,A,\emptyset)$.

We now prove the claim by induction on $|B|+|A|$.

We focus on showing the following part: for all $\omega\in B\cup A$, $\Q(\omega,M_1)=\Q(\omega,M_2)$. A symmetric argument will yield $\Q(\omega,M_1)=\Q(\omega,M_3)$. 

The base cases include $|B|=|A|=0$ and $|B|=|A|=1$. The proof in these cases follows easily.

Thus, we are left with cases where $|B|\geq 1$, $|A|\geq 1$, and $|B|+|A|\geq 3$. Now, we argue that it suffices to consider cases where $|B|\geq 2$. Let us analyze what happens when we run $\UMb$ on $(B,A,M)$, where $|B|=1$ and $|A|\geq 2$. $\UMb$ will compute $B_\succeq = B$ and $A_\succeq = A$ (as $B$ is a singleton and $\vol(B)=\vol(A)$). $\UMb$ checks whether the `median' bid $b$ (the only bid in $B$) and the ask $a$ (the least competitive ask in $A$) are tradable or not. If they are tradable, then each order in $\omega\in B\cup A$ will be exhaustively matched. Since it is a uniform price matching, every bid-ask pair is tradable, so $\match$ on $(B{\downarrow},A{\downarrow},M)$ will also match every order in $B\cup A$ exhaustively, and the claim follows easily. If they are not tradable, then $\UMb$ will return $\UMa(B,A,M)$, i.e., $\UMb(B,A,M)=\UMa(B,A,M)$, and in this case $|A|\geq 2$ and will be handled when we apply the symmetric argument to prove $\Q(\omega,M_1)=\Q(\omega,M_3)$.

Thus, we may assume that $|B|\geq 2$ and this will imply that both $B_\succeq$ and $B_\prec$ will turn out to be proper subsets of $B$ when we run $\UMb$ on $(B,A,M)$.

We fix sets $B$, $A$, and $M$ such that $|B|\geq 2$. Also, $M_1 = \match(B{\downarrow }, A{\downarrow}, M)$ and $M_2= \UMb(B,A, M)$. We need to show that for all $\omega\in B\cup A$, $\Q(\omega,M_1)=\Q(\omega,M_2)$.

In the proof below, we will be using the following facts.

\begin{itemize}
    \item $\UM(B',A') = \match(B'{\downarrow},A'{\downarrow},\emptyset);$
    \item $\match(B',A', M') = M' \cup \match(B',A',\emptyset);$
    \item $\UMb(B',A',M')= M'\cup\match(B',A',\emptyset);$
    \item $\UMa(B',A',M')= M'\cup\UMa(B',A',\emptyset)$, where the union is a disjoint union (which is a list concatenation operation when thinking of the matchings as lists).
\end{itemize}

 $\UMb$ on $(B,A,M)$ first finds the `medians' bid $b$ and ask $a$, and the partitions  $(B_\succeq, B_\prec)$ of $B$, $(A_\succeq, A_\preceq)$ of $A$. Let $q=\vol(B_\succeq)=\vol(A_\succeq)$. $\UMb$ then checks whether $b$ and $a$ are tradable which gives rise to two cases.

 \textbf{Case: $b$ and $a$ are tradable}: In this case, $B_\succeq$ is matched completely with $A_\succeq$ to produce a matching $M'$ with quantity $q$ and the final output matching is $$M_2 = \UMa(B_\prec,A_\preceq,M\cup M')=M\cup M' \cup \UMa(B_\prec,A_\preceq, \emptyset).$$

 Also, $\match(B{\downarrow},A{\downarrow},M)=M\cup\match(B{\downarrow},A{\downarrow},\emptyset)=M\cup\UM(B,A)$.  
We now invoke Lemma~\ref{lem:2} by setting  $\hat b$  to $b$ to argue that the matching output by $\UM$ on $(B,A)$ is\\  $\UM(B_ \succeq, A_\succeq)\cup \UM(B_\prec, A_\preceq)$. Thus, $$M_1 = M\cup \UM(B_ \succeq, A_\succeq)\cup \UM(B_\prec, A_\preceq)= M\cup \UM(B_ \succeq, A_\succeq)\cup \match(B_\prec{\downarrow}, A_\preceq{\downarrow},\emptyset). $$

Note that we have expressed both $M_1$ and $M_2$ as a disjoint union (list concatenation) of three sets. Fix an $\omega\in A\cup B$. Now, $\Q(\omega, M)=\Q(\omega, M)$ (trivially), $\Q(\omega, M')=\Q(\omega, \UM(B_\succeq, A_\succeq))$ (as $M'$ is obtained by exhaustively matching all orders in $(B_ \succeq, A_\succeq)$), and $\Q(\omega, \UMa(B_\prec,A_\preceq, \emptyset)) = \Q(\omega, \match(B_\prec{\downarrow}, A_\preceq{\downarrow},\emptyset))$ (from induction). Thus, we have $\Q(\omega,M_1)=\Q(\omega,M_2)$.

\textbf{Case: $b$ and $a$ are not tradable} $\UMb$ completely discards $B_\prec$ and $A_\preceq$ and outputs the matching $$M_2=\UMa(B_\succeq, A_\succeq, M)=M\cup\UMa(B_\succeq,A_\succeq, \emptyset).$$

Also,
$M_1=\match(B{\downarrow},A{\downarrow},M)= M\cup\match(B{\downarrow},A{\downarrow},\emptyset).$ Observe that no bid in $B_\prec$ is tradable with any ask in $A_\preceq$, as bids in $B_\prec$ are strictly less competitive than $b$ and asks in $A_\preceq$ are at most as competitive as $a$, and $b$ and $a$ are not tradable.
We further claim that $\match(B{\downarrow},A{\downarrow},\emptyset)=\UM(B,A)$ does not match any orders from $B_\prec\cup A_\preceq$.
To see this, we invoke Lemma~\ref{lem:1}. Note that except for bid $b$ and ask $a$, the respective ranges of orders in $B_\succeq$ and $A_\succeq$ have numbers strictly less than $q$ and the respective ranges of orders in $B_\prec$ and $A_\preceq$ have numbers that are all strictly greater than $q$, as $q\in \range_A(a) \cap \range_B(b)$. Thus, any potential matches between $B_\prec$ and $A_\succeq$ or between $A_\preceq$ and $B_\succeq$ can only happen between $b$ and $a$, but they are not tradable (as per the case).
Thus, we conclude that the $\UM(B,A)= \UM(B_\succeq, A_\succeq)$. 
 Therefore, $$M_1= M\cup \match(B_\succeq{\downarrow},A_\succeq{\downarrow},\emptyset). $$

From induction, arguing as before, we get for all orders $\omega\in A\cup B$, $\Q(\omega,M_1)=\Q(\omega,M_2)$, and we are done.
\end{proof}

We now analyze the running time of $\UMs$. Let $T(n)$ represent the running time of $\UMs$, where the number of orders $|B|+|A|=n$.

$\UMs$ calls $\UMb$, which in turn calls $\UMa$ after decreasing the number of the bids by a factor of two. $\UMa$ then calls $\UMb$ again after decreasing the number of the asks by a factor of two. So after two successive returns, we can see that the number of bids and asks decreases by a factor of two. Also, since all the subroutines take linear time, by simple inspection, we conclude

$T(n) \leq T(\lceil \frac n 2 \rceil + 1) + cn$, where $c$ is an absolute constant. Thus, clearly $T(n) = O(n)$.

This completes the proof of Theorem~\ref{thm:result3}.

\section{Conclusions}
\label{sec:conclusions}
The problems we consider are clearly of fundamental interest and we achieve asymptotically tight results for them using elementary techniques. 
Surprisingly, despite their fundamental nature and wide practical applicability, prior to this work, the complexity aspects of such problems were not deeply studied. The following natural questions arise from our work.
\begin{itemize}
    \item 
In this work, the most classical exchange model is assumed; there are execution principles other than price-time priority (like pro-rata matching) which are also being employed in the real world.
These alternative principles present opportunities for studying algorithmic complexity beyond the traditional price-time priority model.
\item Furthermore, it might be interesting to consider similar problems in the context of decentralized exchanges. 
\item Finally, bridging the gap between the upper and lower bounds on the time complexity of Problem 1 in the binary query model remains open.
\end{itemize}

\section*{Acknowledgement}
The authors express their gratitude to Jaikumar Radhakrishnan for bringing the element distinctness problem to their attention and for engaging in insightful discussions that enriched their understanding of it. They are also thankful to him for helping improve the presentation of this work.

\bibliography{main}

\end{document}